\documentclass[11pt]{article}
\usepackage[letterpaper,margin=1in]{geometry}
\usepackage{amsmath, amssymb, amsthm, amsfonts}
\newcommand{\reach}{\Longrightarrow}
\usepackage{ifthen}
\usepackage{arrayjobx}
\usepackage{tikz}
\usetikzlibrary{positioning,decorations.pathreplacing}

\usepackage{cite}
\usepackage{times}
\usepackage{appendix}
\usepackage{graphicx}
\usepackage{color}
\usepackage[noend]{algpseudocode}
\usepackage{epstopdf}
\usepackage{wrapfig}
\usepackage[framemethod=tikz]{mdframed}
\usepackage[bottom]{footmisc}
\usepackage{enumitem}
\setitemize{noitemsep,topsep=0pt,parsep=0pt,partopsep=0pt}
\usepackage{caption}
\usepackage{xspace}
\usepackage{hyperref}
\hypersetup{
    unicode=false,          
    colorlinks=true,        
    linkcolor=red,          
    citecolor=green,        
    filecolor=magenta,      
    urlcolor=cyan           
}
\usepackage[capitalize]{cleveref}
\usepackage{fancyhdr} 
\usepackage{color} 
\usepackage[nofillcomment,linesnumbered,noend,noresetcount,noline]{algorithm2e}
\usepackage[noindentafter]{titlesec}

\titlespacing{\paragraph}{%
  0pt}{
  0.1\baselineskip}{
  1em}
\titlespacing\section{0pt}{8pt plus 1pt minus 1pt}{2pt plus 1pt minus 1pt}
\titlespacing\subsection{0pt}{8pt plus 1pt minus 1pt}{2pt plus 1pt minus 1pt}
\titlespacing\subsubsection{0pt}{8pt plus 1pt minus 1pt}{2pt plus 1pt minus 1pt}
\usepackage{graphicx} 
\usepackage{complexity} 
\usepackage{enumitem} 
\usepackage{shorttoc} 
\usepackage{array} 
\usepackage{float}
\usepackage{pdfsync}
\usepackage{verbatim}
\hypersetup{pageanchor=false}

\AtBeginDocument{%
 \abovedisplayskip=2pt plus 1pt minus 1pt
 \abovedisplayshortskip=0pt plus 1pt
 \belowdisplayskip=2pt plus 1pt minus 1pt
 \belowdisplayshortskip=0pt plus 1pt
}

\newtheoremstyle{slplain}
  {.4\baselineskip\@plus.1\baselineskip\@minus.1\baselineskip}
  {.3\baselineskip\@plus.1\baselineskip\@minus.1\baselineskip}
  {\itshape}
  {}
  {\bfseries}
  {.}
  { }
  {}
\theoremstyle{slplain} 

\algnewcommand\algorithmicswitch{\textbf{switch}}
\algnewcommand\algorithmiccase{\textbf{case}}

\algdef{SE}[SWITCH]{Switch}{EndSwitch}[1]{\algorithmicswitch\ #1\ \algorithmicdo}{\algorithmicend\ \algorithmicswitch}%
\algdef{SE}[CASE]{Case}{EndCase}[1]{\algorithmiccase\ #1}{\algorithmicend\ \algorithmiccase}%
\algtext*{EndSwitch}%
\algtext*{EndCase}%

\renewcommand{\paragraph}[1]{\vspace{0.15cm}\noindent {\bf #1}:}
\newtheorem*{theorem*}{Theorem}
\newtheorem{theorem}{Theorem}[section]
\newtheorem{lemma}[theorem]{Lemma}

\makeatletter
\newtheorem*{rep@theorem}{\rep@title}
\newcommand{\newreptheorem}[2]{%
\newenvironment{rep#1}[1]{%
 \def\rep@title{#2 \ref{##1}}%
 \begin{rep@theorem}}%
 {\end{rep@theorem}}}
\makeatother

\newreptheorem{theorem}{Theorem}
\newreptheorem{lemma}{Lemma}
\newreptheorem{claim}{Claim}
\newreptheorem{invariant}{Invariant}

\theoremstyle{definition}

\theoremstyle{remark}

\numberwithin{equation}{section}

\newtheoremstyle{etplain}
  {.0\baselineskip\@plus.1\baselineskip\@minus.1\baselineskip}
  {.0\baselineskip\@plus.1\baselineskip\@minus.1\baselineskip}
  {\itshape}
  {}
  {\bfseries}
  {.}
  { }
  {}

\newcommand{\idlow}[1]{\mathord{\mathcode`\-="702D\it #1\mathcode`\-="2200}}
\newcommand{\id}[1]{\ensuremath{\idlow{#1}}}
\newcommand{\litlow}[1]{\mathord{\mathcode`\-="702D\sf #1\mathcode`\-="2200}}
\newcommand{\lit}[1]{\ensuremath{\litlow{#1}}}

\newcommand{\namedref}[2]{\hyperref[#2]{#1~\ref*{#2}}}

\newcommand{\theoremref}[1]{\namedref{Theorem}{#1}}

\newcommand{\figureref}[1]{\namedref{Figure}{#1}}
\newcommand{\figurerefb}[2]{\hyperref[#1]{Figure~\ref*{#1}#2}}

\newcommand{\lemmaref}[1]{\namedref{Lemma}{#1}}

\newcommand{\equationref}[1]{\hyperref[#1]{(\ref*{#1})}}
\renewcommand{\eqref}{\equationref}

\usepackage[textsize=tiny]{todonotes}

\newcommand{\DEBUG}[1]{}

\newcommand{\EX}{\operatornamewithlimits{\mathbb{E}}}

\newcommand{\FullOrShort}{short}

\ifthenelse{\equal{\FullOrShort}{full}}{
        
  \newcommand{\fullOnly}[1]{#1}
  \newcommand{\shortOnly}[1]{}
  
  }{

    \newcommand{\fullOnly}[1]{}
    \newcommand{\shortOnly}[1]{#1}
    
  }

\newcommand{\LM}[1]{\emph{LM}}

\begin{document}

\date{}

\title{Polylogarithmic-Time Leader Election  in Population Protocols }

\author{
 Dan Alistarh\\
  \small Microsoft Research\\
	\small daalista@microsoft.com\\
  \and
 Rati Gelashvili\footnote{Work performed in part while an intern with Microsoft Research.}\\
  \small MIT\\
	\small gelash@mit.edu
}

\maketitle
\begin{abstract}

Population protocols are networks of finite-state agents, interacting randomly, and updating their states using simple rules. 
Despite their extreme simplicity, these systems have been shown to cooperatively perform complex computational tasks, such as simulating register machines to compute standard arithmetic functions.
The election of a unique \emph{leader agent} is a key requirement in such  computational constructions. Yet, the fastest currently known population protocol for electing a leader only has \emph{linear} stabilization time, and it has recently been shown that no population protocol using a \emph{constant} number of states per node may overcome this linear bound. 

In this paper, we give the first population protocol for leader election with \emph{polylogarithmic} stabilization time, using polylogarithmic memory states per node. 
The protocol structure is quite simple: each node has an associated value, and is either a \emph{leader} (still in contention) or a \emph{minion} (following some leader). 
A leader keeps incrementing its value and ``defeats'' other leaders in one-to-one interactions, 
  and will drop from contention and become a minion if it meets a leader with higher value. 
Importantly, a leader also drops out if it meets a \emph{minion} with higher absolute value. 
While these rules are quite simple, the proof that this algorithm achieves polylogarithmic stabilization time is non-trivial. 
In particular, the argument combines careful use of concentration inequalities with anti-concentration bounds, showing that the leaders' values become spread apart as the execution progresses, which in turn implies that straggling leaders get quickly eliminated.  
We complement our analysis with empirical results, 
  showing that our protocol stabilizes extremely fast, even for large network sizes. 
\end{abstract}

\setcounter{page}{1}

\section{Introduction}
Recently, there has been significant interest in modeling and analyzing interactions arising in 
biological or bio-chemical systems through an algorithmic lens. 
Several interesting computational models have been proposed for such networks, for example, the cellular automata model~\cite{CA}, the stone-age distributed computing model~\cite{EW13}, or the population model~\cite{AADFP06}.

In particular, population protocols~\cite{AADFP06}, which are the focus of this paper, consist of a set of $n$ finite-state nodes interacting in pairs, where each interaction may update the states of both participants. The goal is to have all nodes stabilize on an output value, which represents the result of the computation, usually a predicate on the initial state of the nodes. 
The set of interactions occurring at each step is assumed to be decided by an adversarial scheduler, which is usually subject to some fairness conditions. 
The standard scheduler is the \emph{probabilistic (uniform random) scheduler}~\cite{AAE08, PVV09, DV12, Spirakis}, which picks the next pair to interact uniformly at random in each step. We adopt this probabilistic scheduler model in this paper. (Some references refer to this model as the \emph{probabilistic} population model.)
The fundamental measure of stabilization is \emph{parallel time}, defined as the number of scheduler steps until stabilization, divided by $n$.\footnote{An alternative definition is when reactions occur in parallel according to a Poisson process~\cite{PVV09, DV12}.}

The class of predicates computable by population protocols is now well-understood~\cite{AADFP06, AngluinAE2006, AngluinAER2007} to consist precisely of \emph{semilinear predicates}, i.e. predicates definable in first-order Presburger arithmetic. The first such construction was given in~\cite{AADFP06}, and later improved in terms of convergence time in~\cite{AAE08le}. 
A parallel line of research studied the computability of deterministic functions in chemical reaction networks, which are also instances of population protocols~\cite{chen2014deterministic}. 
All three constructions fundamentally rely on the election of a single initial \emph{leader} node, which co-ordinates phases of computation.
 
Reference~\cite{AADFP06} gives a simple protocol for electing a leader from a uniform population, 
based on the natural idea of having leaders eliminate each other directly through symmetry breaking. 
Unfortunately, this strategy takes at least linear parallel time in the number of nodes $n$: 
for instance, once this algorithm reaches \emph{two} surviving leaders, it will require $\Omega( n^2 )$ additional interactions for these two leaders to meet. 
Reference~\cite{AAE08le} proposes a significantly more complex protocol, conjectured to be sub-linear, and whose convergence is only studied experimentally. 
These references posit the existence of a sublinear-time population protocol for leader election as a ``pressing'' open problem in the area. In fact, the existence of a poly-logarithmic leader election protocol would imply that \emph{any semilinear predicate} is computable in poly-logarithmic time by a \emph{uniform population}~\cite{AAE08le}. 

Recently, Doty and Soloveichik~\cite{DS15} showed that $\Omega( n^2 )$ expected  interactions are \emph{necessary} for electing a leader in the classic probabilistic protocol model in which each node only has \emph{constant} number of memory states (with respect to $n$). This negative result implies that computing semilinear predicates in leader-based frameworks is subject to the same lower bound. In turn, this motivates the question of whether faster computation is possible if the amount of memory per node is allowed to be a function of $n$. 

\paragraph{Contribution} 
In this paper, we solve this problem by proposing a new population protocol for leader election, 
  which stabilizes in $O( \log^3 n )$ expected parallel time, using $O( \log^3 n)$ memory states per node. 
Our protocol, called \LM, for \emph{Leader-Minion}, roughly works as follows. 
Throughout the execution, each node is either a \emph{leader}, meaning that it can still win, or a \emph{minion}, following some leader. 
Each node state is associated to some \emph{absolute value}, which is a positive integer, 
  and with a \emph{sign}, positive if the node is still in contention, and negative if the node has become a minion. 

If two leaders meet, the one with the larger absolute value survives, and increments its value, while the other drops out, 
  becoming a minion, and adopting the other node's value, but with a negative sign. 
  (If both leaders have the same value, they both increment it and continue.) 
If a leader meets a minion with \emph{smaller} absolute value than its own, it increments its value, 
  while the minion simply adopts the leader's value, but keeps the negative sign. 
Conversely, if a leader meets a minion with \emph{larger} absolute value than its own, 
  then the leader drops out of contention, adopting the minion's value, with negative sign. 
Finally, if two minions meet, they update their values to the maximum absolute value between them, but with a \emph{negative} sign. 

These rules ensure that, eventually, a single leader survives. While the protocol is relatively simple, the proof of poly-logarithmic time stabilization is non-trivial. 
In particular, the efficiency of the algorithm hinges on the minion mechanism, which ensures that a leader with high absolute value can eliminate other contenders in the system,  without having to directly interact with them. 

Roughly, the argument is based on two technical insights. 
First, consider two leaders at a given time $T$, 
  whose (positive) values are at least $\Theta( \log n )$ apart. 
Then, we show that, within $O( \log n )$ parallel time from $T$, 
  the node holding the smaller value has become a minion, with constant probability. 
Intuitively, this holds since 
  1) this node will probably meet either the other leader or one of its minions within this time interval, and 
  2) it cannot increase its count fast enough to avoid defeat.  
For the second part of the argument, we show via anti-concentration that, 
  after parallel time $\Theta( \log^2 n )$ in the execution, 
  the values corresponding to an arbitrary pair of nodes will be separated by at least $\Omega( \log n )$.  

We ensure that the values of nodes cannot grow beyond a certain threshold, 
  and set the threshold in such a way that the total number of states is $\Theta(\log^3 n)$.
We show that with high probability the leader will be elected before the values of the nodes reach the threshold.
In the other case, remaining leaders with threshold values engage in a \emph{backup} dynamics where minions are irrelevant 
  and leaders defeat each other when they meet based on random binary indicators 
  which are set using the randomness of the scheduler.
This process is slower but correct and only happens with very low probability, allowing to conclude that 
 the \LM{} algorithm stabilizes to a single leader 
  within $O( \log^3 n )$ parallel time, both with high probability and in expectation,
  using $O( \log^3 n )$ states.

In population protocols, in every interaction, 
  one node is said to be the \emph{initiator}, 
  the other is the \emph{responder},
  and the state update rules can use this distinction.
In our protocol, this would allow a leader (the initiator in the interaction) 
  to defeat another leader with the same value (the responder),
  and could also simplify the backup dynamics of our algorithm. 
However, our \LM{} algorithm has the nice property that the state update rules can be made completely symmetric
  with regards to the initiator and responder roles.\footnote{For this reason, \LM{} algorithm works for $n>2$ nodes, because to elect a leader among two nodes it is necessary to rely on the initiator-responder role distinction.} 
     
Summing up, we give the first poly-logarithmic time protocol for electing a leader from a uniform population. 
We note that $\Omega( n \log n )$ interactions seem intuitively necessary for leader election, as this number is required to allow each node to interact at least once. 
However, this idea fails to cover all possible reaction strategies if nodes are allowed to have arbitrarily many states. 

 We complement our analysis with empirical data, suggesting that the stabilization time of our protocol is close to logarithmic, and that in fact the asymptotic constants are small, both in the stabilization bound, and in the upper bound on the number of states the protocol employs.

\paragraph{Related Work}
We restrict our attention to work in the population model. 
The framework of population protocols was formally introduced in reference~\cite{AADFP06}, to model interactions arising in biological, chemical, or  sensor networks. It sparked research into its computational power~\cite{AADFP06, AngluinAE2006, AngluinAER2007}, and into the time complexity of fundamental tasks such as majority~\cite{AAE08, PVV09, DV12, Spirakis, AGVsubmitted}, and leader election~\cite{AADFP06, AngluinAE2006}.\footnote{The best known upper bound for deterministic majority is of $O( \log n \log s + \log n / (\epsilon s))$ parallel time~\cite{AGVsubmitted}, where $n$ is the number of nodes, $s$ is the number of states per node, and $\epsilon$ is the initial node difference between the two input states. The two problems are complementary, and no complexity-preserving transformations exist, to our knowledge. 
}
References interested in \emph{computability} consider an adversarial scheduler which is restricted to be \emph{fair}, e.g., where each agent interacts with every other agent infinitely many times. 
For complexity bounds, the standard scheduler is \emph{uniform}, scheduling each pair uniformly at random at each step, e.g.,~\cite{AAE08, PVV09, DV12, Spirakis}. This model is also known as the \emph{probabilistic} population model.

To the best of our knowledge, no population protocol for electing a leader with sub-linear stabilization
  time was known before our work. 
References~\cite{AADFP06, AngluinAE2006, chen2014deterministic} present 
leader-based frameworks for population computations, assuming the existence of such a node. 
The existence of such a sub-linear protocol is stated as an open problem in~\cite{AADFP06, AngluinAE2006}. Reference~\cite{doty2013leaderless} proposes a \emph{leader-less} framework for population computation. 

Recent work by Doty and Soloveichik~\cite{DS15} showed an $\Omega( n^2 )$ lower bound on the number of interactions necessary for electing a leader in the classic probabilistic protocol model in which each node only has \emph{constant} number of memory states with respect to the number of nodes $n$~\cite{AngluinAER2007}. 
The proof of this result is quite complex, and makes use of the limitation 
that the number of states remains constant even as the number of nodes $n$ is taken to tend to infinity. 

Thus, our algorithm can be interpreted as a complexity separation between population protocols which may only use constant memory per node, and protocols where the number of states is allowed to be a function of $n$. 

A parallel line of research studied \emph{self-stabilizing} population protocols, e.g., ~\cite{angluin2006self, fischer2006self, sudo2010loosely}, that is, protocols which can stabilize to a correct solution from an arbitrary initial state. 
It is known that stable leader election is impossible in such systems~\cite{angluin2006self}; references~\cite{fischer2006self, sudo2010loosely}
circumvent this impossibility by relaxing the problem semantics. 
Our algorithm is not affected by this result since it is not self-stabilizing.
\section{Preliminaries}
\paragraph{Population Protocols} 
We assume a population consisting of $n$ agents, or nodes, each executing as a deterministic state machine with states from a finite set $Q$, with a finite set of input symbols $X \subseteq Q$, a finite set of output symbols $Y$, a transition function $\delta : Q \times Q \rightarrow Q \times Q$, and an output function $\gamma : Q \rightarrow Y$. 
Initially, each agent starts with an input from the set $X$, and proceeds to update its state following interactions with other agents, according to the transition function $\delta$. 

Agents are \emph{anonymous}, 
  so any two agents in the same state are identical and interchangeable. 
Thus, we represent any set of agents simply 
  by the \emph{counts of agents} in every state, which we call a \emph{configuration}.
More formally, a \emph{configuration} $c$ is a function 
  $c: Q \to \mathbb{N}$, where $c(S)$ represents the 
  \emph{number of agents in state $S$ in configuration $c$}.
We let $|c|$ stand for the sum, over all states $S \in Q$, of $c(S)$,
  which is the same as the total number of agents in configuration $c$.
For instance, if $c$ is a configuration of all agents in the system,
  then $c$ describes the global state of the system, and $|c| = n$. 
We say that a configuration $c'$ is \emph{reachable} from a configuration $c$, 
  denoted $c \reach c'$, if there exists a sequence of consecutive steps 
  (interactions from $\delta$ between pairs of agents) 
  leading from $c$ to $c'$.

The agents' interactions proceed according to a directed \emph{interaction graph} $G$ without self-loops, whose edges indicate possible agent interactions. 
Usually, the graph $G$ is considered to be the complete graph on $n$ vertices, a convention we also adopt in this paper. 

The execution proceeds in \emph{steps}, or \emph{rounds}, where in each step a new edge $(u, w)$ is chosen uniformly at random from the set of edges of $G$. Each of the two chosen agents updates its state according to function $\delta$. 


\paragraph{Parallel Time} The above setup considers sequential interactions; however, in general, interactions between pairs of distinct agents are independent, and are usually considered as occurring in parallel. In particular, it is customary to define one unit of \emph{parallel time} as $n$ consecutive steps of the protocol. 

\paragraph{The Leader Election Problem} 
In the \emph{leader election} problem, all agents start in the same initial state $A$, 
  i.e. the only state in the input set $X = \{A\}$. 
The output set is $Y = \{\id{Win}, \id{Lose}\}$. 

We say that a configuration $c$ \emph{has a single leader} if 
  there exists some state $S \in Q$ with $\gamma(S) = \id{Win}$ and $c(S) = 1$, 
  such that for any other state $S' \neq S$, $c(S') > 0$ implies $\gamma(S') = \id{Lose}$.
A configuration $c$ of $n$ agents has a \emph{stable leader},
  if for all $c'$ reachable from $c$, it holds that $c'$ has a single leader.

A population protocol \emph{stably elects a leader}
  within $\ell$ steps with probability $1 - \phi$, 
  if, with probability $1 - \phi$, any configuration $c$ reachable 
  by the protocol after $\geq \ell$ steps has a stable leader.


\section{The Leader Election Algorithm}
In this section, we describe the \LM{} leader election algorithm. 
The algorithm has an integer parameter $m > 0$, which we set to $\Theta(\log^3 n)$. 
Each state corresponds to an integer value from the set 
  $\{-m, -m+1, \ldots, -2, -1, 1, 2, m-1, m, m+1\}$. 
Respectively, there are $2m+1$ different states.
We will refer to states and values interchangeably.
All nodes start in the same state corresponding to value $1$.

The algorithm, specified in~\figureref{fig:lepp}, 
  consists of a set of simple deterministic update rules for the node state. 
In the pseudocode, the node states before an interaction are denoted by $x$ and $y$, 
  while their new states are given by $x'$ and $y'$.
All nodes start with value $1$ and continue to interact according to these simple rules. 
We prove that all nodes except one will stabilize to negative values, 
  and that stabilization is fast with high probability. 
This solves the leader election problem since 
  we can define $\gamma$ as mapping only positive states to $\id{Win}$ (a leader). 
  (Alternatively, $\gamma$ that maps only two states with values $m$ and $m+1$ to $\id{WIN}$ would also work, 
  but we will work with positive leader states for the simplicity of presentation.)

Since positive states translate to being a leader according to $\gamma$,
  we call a node a \emph{contender} if it has a positive value, and a \emph{minion} otherwise.
We present the algorithm in detail below.
\begin{figure}[ht]
\hrule
\DontPrintSemicolon
{\centering
{\small
\begin{algorithm}[H]
\SetKwInput{KwState}{Parameters}
\KwState{\;
$m,$ an integer $>0$, set to $\Theta(\log^3{n})$\;
}
\SetKwInput{KwState}{State Space}
\KwState{\;
$\id{Leader States} = \{ 1, 2, \ldots, m - 1, m, m + 1 \}$, \;
$\id{Minion States} = \{ -1, -2 , \ldots, -m + 1, -m \}$,\;
}

\KwIn{States of two nodes, $x$ and $y$}
\KwOut{Updated states $x'$ and $y'$}

\SetKwInput{KwState}{Auxiliary Procedures}

\KwState{\;
$\id{is-contender}( x ) = \left\{ 
 \begin{array}{ll} 
  \lit{true}  & \textnormal{if $x \in \id{Leader States}$; } \\
  \lit{false} & \textnormal{otherwise.}
  \end{array} 
  \right. $

$\id{contend-priority}( x, y ) = \left\{ 
 \begin{array}{ll} 
  m                  & \textnormal{if $\max(|x|, |y|) = m + 1$; } \\
  \max(|x|, |y|) + 1 & \textnormal{otherwise.}
  \end{array} 
  \right. $

$\id{minion-priority}( x, y ) = \left\{ 
 \begin{array}{ll} 
  -m              & \textnormal{if $\max(|x|, |y|) = m + 1$; } \\
  -\max(|x|, |y|) & \textnormal{otherwise.}
  \end{array} 
  \right. $
}

\BlankLine
\textbf{procedure} $\lit{update}\langle x, y \rangle$\;
{
\Indp

\If{$\id{is-contender}(x)$ \textbf{and} $|x| \geq |y|$ \nllabel{line:xst}}
{
  $x' \gets \id{contend-priority}(x, y)$ \nllabel{line:xm}
} \lElse {$x' \gets \id{minion-priority}(x, y)$ \nllabel{line:xen}}

\If{$\id{is-contender}(y)$ \textbf{and} $|y| \geq |x|$ \nllabel{line:yst}}
{
  $y' \gets \id{contend-priority}(x, y)$ \nllabel{line:ym}
} \lElse {$y' \gets \id{minion-priority}(x, y)$ \nllabel{line:yen}}

\Indm
}
\end{algorithm}}}
\hrule
\caption{The state update rules for the \LM{} algorithm.}
\label{fig:lepp}
\end{figure}

The state updates (i.e. the transition function $\delta$) of the \LM{} algorithm  
  are completely symmetric, that is, 
  the new state $x'$ depends on $x$ and $y$ (lines~\ref{line:xst}-\ref{line:xen}) 
  exactly as $y'$ depends on $y$ and $x$ (lines~\ref{line:yst}-\ref{line:yen}). 

If a node is a contender and has absolute value not less than the absolute value of the interaction partner,
  then the node remains a contender and updates its value using the \emph{contend-priority} function 
  (lines~\ref{line:xm} and~\ref{line:ym}).
The new value will be one larger than the previous value except when the previous value was $m+1$, 
  in which case the new value will be $m$.

If a node had a smaller absolute value than its interaction partner, or was a minion already,
  then the node will be a minion after the interaction.
It will set its value using the \emph{minion-priority} function,  
  to either $-\max(|x|, |y|)$, or $-m$ if the maximum was $m+1$ (lines~\ref{line:xen} and~\ref{line:yen}). 

Values $m+1$ and $m$ are treated exactly the same way by minions (essentially corresponding to $-m$).
These values serve as a binary tie-breaker 
  among the contenders that ever reach the value $m$, as will become clear from the analysis.
\section{Analysis}
In this section, we provide a complete analysis of our leader election algorithm. 

\paragraph{Notation} 
Throughout the proof, 
we call a node \emph{contender} when the value associated with its state is positive,
  and a \emph{minion} when the value is negative. 
As previously discussed, we assume that $n>2$. 
For presentation purposes, we also consider $n$ to be a power of two.
We measure execution time in discrete steps (rounds), where each step corresponds to an interaction. 

We first prove that the algorithm never eliminates all contenders and 
  that a configuration with a single contender means that a leader is elected.
\begin{lemma}
\label{lem:correct}
There is always at least one contender in the system.
Suppose the execution reaches a configuration $c$ with only node $v$ being a contender.
Then, $v$ remains a contender (mapped to $\id{WIN}$ by $\gamma$) 
  in any configuration $c'$ reachable from $c$,
  and $c'$ never contains another contender.
\end{lemma}
\begin{proof}
By the structure of the algorithm, 
  a node starts as a contender and may become a minion during an execution, 
  but a minion may never become a contender.
Moreover, an absolute value associated with the state of a minion node can only increase to 
  an absolute value of an interaction partner.

Suppose for contradiction that an execution reaches 
  a configuration $\hat{c}$ where all nodes are minions.
Let the maximum absolute value of the nodes be $u$ in $\hat{c}$.
Because the minions cannot increase the maximum absolute value in the system, 
  there must have been a contender with value $u$ during the execution 
  before the execution reached $\hat{c}$.
For this contender to have become a minion, 
  it must have interacted with another node with an absolute value strictly larger than $u$.
The absolute value of a node never decreases except from $m+1$ to $m$,
  and despite existence of a larger absolute value than $u$ before reaching $\hat{c}$, 
  $u$ was the largest absolute value in $\hat{c}$.
Thus, $u$ must be equal to $m$. 
But after such an interaction, the second node that was in the state $m+1$ 
  remains a contender with value $m$.
Before the execution reaching $\hat{c}$, 
  it must also have interacted with yet another node with value $m+1$
  in order to become a minion itself.
But then, the interaction partner remains a contender with value $m$ 
  and the same reasoning applies to it.
Our proof follows by infinite descent.

Consequently, whenever there is a single contender in the system, 
  it must have the largest absolute value.
Otherwise, it could interact with a node with a larger absolute value and become a minion,
  contradicting the above proof that all nodes may never be minions.
Due to this invariant, the only contender may never become a minion 
  and we know the minions can never become contenders.
\end{proof}
Now we turn our attention to the stabilization speed (assuming $n > 2$) of the \LM{} algorithm.
Our goal is bound the number of steps necessary to eliminate all except a single contender.
In order for a contender to get eliminated, 
  it must come across a larger value of another contender, 
  the value possibly conducted through a chain of multiple minions via multiple interactions.

We first show by a rumor spreading argument that if the difference between the 
  values of two contenders is large enough, then the contender with the smaller value will 
  become a minion within the next $O(n \log n)$ interactions, with constant probability.
Then we use anti-concentration bounds to establish that for any two fixed contenders, 
  given that no absolute value in the system reaches $m$, after every $O(n \log^2 n)$ 
  interactions the difference between their values is large enough with constant probability.
\begin{lemma}
\label{lem:rumor}
Consider a configuration $c$, in which there are two contenders with values $u_1$ and $u_2$, 
  where $ u_1 - u_2 \geq  4 \xi \log n$ for $\xi \geq 8$.
Then, after $\xi n \log n$ interactions from $c$, 
  the node that initially held the value $u_2$ will be a minion with probability at least $1/24$
  (independent of the history of previous interactions leading up to $c$).
\end{lemma}
\begin{proof}
We call a node that has an absolute value of at least $u_1$ an \emph{up-to-date} node, 
  and \emph{out-of-date} otherwise.
Initially, at least one node is up-to-date.
When there are $x$ up-to-date nodes,
  the probability that an out-of-date node interacts with an up-to-date node next,
  increasing the number of up-to-date nodes to $x+1$, is $\frac{2x(n-x)}{n(n-1)}$.
By a Coupon Collector argument, 
  the expected number of steps until every node is up-to-date is
  $\sum_{x = 1}^{n-1}\frac{n (n-1)}{2x(n-x)}\leq\frac{(n-1)}{2}\sum_{x = 1}^{n-1}\left(\frac{1}{x}+\frac{1}{n-x}\right)\leq 2 n\log n$.
  
By Markov's inequality, the probability that not all nodes are up-to-date 
  after $\xi n \log n$ interactions is at most $2 / \xi$. 
Hence, expected number of up-to-date nodes after $\xi n \log n$ interactions
  is at least $\frac{ n (\xi - 2)}{\xi}$.
Let $q$ be the probability that the number of up-to-date nodes after $\xi n \log n$ interactions 
  is at least $\frac{ n }{ 3} + 1$. 
We have $qn + (1-q) (\frac{ n }{ 3} + 1) \geq \EX[Y] \geq \frac{ n (\xi - 2)}{\xi}$,
  which implies $q \geq \frac{1}{4}$ for $n>2$ and $\xi \geq 8$.

Hence, with probability at least $1/4$, at least $n/3 + 1$ are nodes are up to date after 
  $\xi n \log n$ interactions from configuration $c$.
By symmetry, the $n/3$ up-to-date nodes except the original node are uniformly 
  random among the other $n-1$ nodes. 
Therefore, any given node, in particular the node that had value $u_2$ in $c$
  has probability at least $1/4 \cdot 1/3 = 1/12$
  to be up-to-date after $\xi n \log n$ interactions.
When the node that was holding value $u_2$ in $c$ becomes up-to-date and 
  gets an absolute value of at least $u_1$ from an interaction,
  it must become a minion by the structure of the algorithm 
  if its value before this interaction was still strictly smaller than $u_1$.   
Thus, we only need to show that the probability of selecting the node that initially had value $u_2$ 
  at least $4 \xi \log n$ times (so that its value can reach $u_1$) 
  during these $\xi n \log n$ interactions is at most $1/24$.
The claim then follows by Union Bound.

In each interaction, the probability to select this node (that initially held $u_2$) is $2/n$.
Let us describe the number of times it is selected in $\xi n \log{n}$ interactions 
  by considering a random variable $Z \sim \mathrm{Bin}(\xi n \log{n}, 2/n)$. 
By Chernoff Bound, the probability being selected at least $4 \xi \log n$ times is at most:
\begin{align*}
\Pr\left[Z \geq 4 \xi \log n \right] 
  \leq \exp \left( -\frac{2\xi}{3} \log{n} \right) \leq \frac{1}{n^{2\xi/3}} \leq \frac{1}{24}
\end{align*} 
\noindent finishing the proof.
\end{proof}
Next, we show that, after $\Theta( n \log^2 n )$ interactions,
  the difference between the values of any two given contenders is high, 
  with a reasonable probability. 
\begin{lemma}
\label{lem:gap}
For an arbitrary configuration $c$, fix two conteders in $c$
  and a constant $\xi \geq 1$. 
Let $c'$ be a configuration reached after $32 \xi^2 n \log^2 n$ interactions from $c$.

If absolute values of all nodes are strictly less than $m$ at all times before reaching $c'$, then, 
  with probability at least $\frac{1}{24} - \frac{1}{n^{8\xi}}$, in $c'$,
  either at least one of the two fixed nodes have become minions,
  or their absolute values differ by at least $4 \xi \log n$. 
\end{lemma}
\begin{proof}
Suppose no absolute value reaches $m$ at any point before reaching $c'$
  and that the two fixed nodes are still contenders in $c'$.
We need to prove that the difference of values is large enough. 

Consider the $32 \xi^2 n \log^2 n$ interactions following $c$.
If an interaction involves exactly one of the two fixed nodes, we call it a \emph{spreading}.
For each interaction, probability of it being spreading is $\frac{4(n-2)}{n(n-1)}$,
  which for $n > 2$ is at least $2/n$.
So, we can describe the number of spreading interactions among the $32 \xi^2 n \log^2{n}$ steps 
  by considering a random variable $X \sim \mathrm{Bin}(32 \xi^2 n \log^2{n}, 2/n)$.
By Chernoff Bound, the probability of having no more than $32 \xi^2 \log^2{n}$ 
  spreading interactions is at most
\begin{align*}
\Pr\left[X \leq 32 \xi^2 \log^2{n} \right] 
\leq \exp\left(- \frac{64 \xi^2\log^2{n}}{2^2\cdot2} \right) 
< \frac{1}{n^{8\xi}},
\end{align*}
Let us from now on focus on the high probability event that there are 
  at least $32 \xi^2 \log^2{n}$ spreading interactions between $c$ and $c'$, 
  and prove that the desired difference will be large enough with probability $\frac{1}{24}$.
This implies the claim by Union Bound with the above event 
  (since for $n>2$, $\frac{1}{n^{8\xi}} < \frac{1}{24}$ holds). 

We assumed that both nodes remain contenders up until $c'$.
Hence, in each spreading interaction, a value of exactly one of them, 
  with probability $1/2$ each, increases by one.
Let us call the fixed nodes $V_1$ and $V_2$,
  and suppose the value of $V_1$ was not less than the value of $V_2$ in $c$.
Let us now focus on the sum $Y$ of $k$ independent uniformly distributed $\pm 1$ 
  Bernoulli trials $x_i$ with $1 \leq i \leq k$,
  where each trial corresponds to a spreading interaction and 
  outcome $+1$ means that the value of $V_1$ increased,
  while $-1$ means that the value of $V_2$ increased.
In this terminology, we are done if we show that $\Pr[ Y \geq 4 \xi \log n] \geq \frac{1}{24}$ 
  for $k \geq 32 \xi^2 \log^2{n}$ trials.

However, we have that:
\begin{align}
\Pr[ Y \geq 4 \xi \log n] &\geq \frac{\Pr[|Y| \geq 4\xi \log n]}{2} = \frac{\Pr[|Y^2| \geq 16 \xi^2 \log^2 n]}{2} \label{eq:sym}\\
 &\geq \frac{\Pr[ |Y^2| \geq k/2]}{2} = \frac{\Pr[ |Y^2| \geq \EX[Y^2]/2]}{2} \label{eq:ksdef}\\
 &\geq \frac{1}{2^2 \cdot 2} \frac{ \EX[Y^2]^2 }{ \EX[Y^4] } \geq \frac{1}{24} \label{eq:exps}
\end{align} 
\noindent where~\ref{eq:sym} follows from the symmetry of the sum with regards to the sign, 
  that is, from
  $\Pr[ Y > 4 \xi \log n] = \Pr[ Y < -4 \xi \log n]$. 
For~\ref{eq:ksdef} we have used that $k \geq 32 \xi^2 \log^2{n}$ and $\EX[Y^2] = k$ 
  (more about this below).
Finally, to get~\ref{eq:exps} we use Paley-Zygmund inequality and 
  the fact that $\EX[Y^4] = 3k(k-1) + k \leq 3k^2$.
Evaluating $\EX[Y^2]$ and $\EX[Y^4]$ is simple by using the definition of $Y$ and 
  the linearity of expectation.
The expectation of each term then is either $0$ or $1$ and it suffices 
  to count the number of terms with expectation $1$,
  which are exactly the terms where each multiplier is raised to an even power.    
\end{proof}
We are ready to prove the stabilization speed with high probability
\begin{theorem}
\label{thm:whp}
There exists a constant $\alpha$, such that for any constant $\beta \geq 3$ following holds:
If we set $m = \alpha \beta \log^3 n = \Theta(\log^3 n)$,
  the algorithm elects a leader (i.e. reaches a configuration with a single contender) 
  in at most $O(n \log^3 n)$ steps, i.e. in parallel time $O(\log^3 n)$,
  with probability at least $1 - 1/n^{\beta}$.  
\end{theorem}
\begin{proof}
Let us fix $\xi \geq 8$ large enough, such that for some constant $p$
\begin{equation}
\label{eq:gammab}
\frac{1}{24} \cdot \left(\frac{1}{24} - \frac{1}{n^{8\xi}}\right) \geq p.
\end{equation}
Consider constants $\beta \geq 3$ and $\alpha = \frac{16}{p} \cdot (33 \xi^2)$.
We set $m = \alpha \beta \log^3 n$ and focus on the first $\frac{\alpha \beta n \log^3 n}{4}$ 
  steps of the algorithm execution.
For any fixed node, the probability that it interacts in each step is $2/n$.
Let us describe the number of times a given node interacts within the first 
  $\frac{\alpha \beta n \log^3 n}{4}$ steps
  by considering a random variable $\mathrm{Bin}(\frac{\alpha \beta n \log^3 n}{4}, 2/n)$. 
By Chernoff Bound, the probability being selected at least $m = \alpha \beta \log^3 n$ 
  times is at most
$\exp \left( -\frac{\alpha\beta}{6} \log^3{n} \right) \leq \frac{1}{n^{\alpha\beta/6}} $.
By Union Bound over all $n$ nodes, with probability at least $1-\frac{n}{n^{\alpha \beta/6}}$,
  all nodes interact strictly less than $m$ times 
  during the first $\frac{\alpha \beta n \log^3 n}{4}$ interactions.

Let us from now on focus on the above high probability event, 
  which means that all absolute values are strictly less than $m$ during the first 
  $\frac{\alpha \beta n \log^3 n}{4} = \frac{4 \beta}{p} (33 \xi^2) n \log^3 n$ interactions.
For a fixed pair of nodes, we apply~\lemmaref{lem:gap} followed by~\lemmaref{lem:rumor} 
  (with parameter $\xi$) 
  $\frac{4 \beta (33 \xi^2) n \log^3 n}{p(32 \xi^2 n \log^2 n + \xi n \log n)} \geq \frac{4 \beta \log{n}}{p}$ times.
Each time, by~\lemmaref{lem:gap}, after $32 \xi^2 n \log^2 n$ interactions 
  with probability at least $\frac{1}{24} - \frac{1}{n^{8\xi}}$ the nodes end up with values 
  at least $4 \xi \log n$ apart.
In this case, after the next $\xi n \log n$ interactions, by~\lemmaref{lem:rumor}, 
  one of the nodes becomes a minion with probability at least $1/24$. 
Since~\lemmaref{lem:rumor} is independent from the interactions that precede it,
  by~\equationref{eq:gammab}, each of the $\frac{4 \beta \log{n}}{p}$ times 
  if both nodes were contenders, with probability at least $p$ one of the nodes becomes a minion.
The probability that both nodes in a given pair are still contenders after the first 
  $\frac{\alpha \beta n \log^3 n}{4}$ steps is thus at most
  $(1-p)^{\frac{4 \beta \log{n}}{p}} \leq 2^{-4 \beta \log{n}} < \frac{1}{n^{2 \beta}}$.
By Union Bound over all $\frac{n(n-1)}{2} < n^2$ pairs,
  with probability at least $1-\frac{n^2}{n^{2\beta}}$, for every pair of nodes,
  one of them is a minion after $\frac{\alpha \beta n \log^3 n}{4}$ interactions.
Hence, with this probability, there will be only one contender.

Combining with the conditioned event that none of the nodes interact $m$ or more times
  gives that after the first $\frac{\alpha \beta n \log^3 n}{4} = O(n \log^3 n)$ interactions
  there must be a single contender with probability at least 
  $1 - \frac{n^2}{n^{2\beta}} - \frac{n}{n^{\alpha\beta/6}} \geq 1 - \frac{1}{n^{\beta}}$ 
  for $\beta \geq 3$.
A single contender means that leader is elected by~\lemmaref{lem:correct}.
\end{proof}
Finally, we prove the expected stabilization bound
\begin{theorem}
\label{thm:exp}
There is a setting of parameter $m$ of the algorithm such that $m = \Theta(\log^3 n)$,
  such that the algorithm elects the leader in expected $O(n \log^3 n)$ steps,
  i.e. in parallel time $O(\log^3 n)$.
\end{theorem}
\begin{proof}
Let us prove that from any configuration, 
  the algorithm elects a leader in expected $O(n \log^3 n)$ steps.
By~\lemmaref{lem:correct}, there is always a contender in the system and 
  if there is only a single contender, then a leader is already elected.
Now in a configuration with at least two contenders consider any two of them.
If their values differ, then with probability at least $1/n^2$ these two contenders will 
  interact next and the one with the lower value will become a minion 
  (after which it may never be a contender again).
If the values are the same, then with probability at least $1/n$, 
  one of these nodes will interact with one of the other nodes, 
  leading to a configuration where the values of our two nodes differ\footnote{This is always true, even when the new value is not larger, for instance when the values were equal to $m+1$, the new value of one of the nodes will be $m \neq m+1$.},
  from where in the next step, independently, 
  with probability at least $1/n^2$ these nodes will interact and one of them will become a minion.
Hence, unless a leader is already elected, in every two steps, 
  with probability at least $1/n^3$ the number of contenders decreases by $1$.

Thus, the expected number of interactions until the number of contenders decreases by $1$ 
  is at most $2 n^3$.
In any configuration there can be at most $n$ contenders, 
  thus the expected number of interactions until reaching a configuration with only 
  a single contender is at most $2 (n-1) n^3 \leq 2n^4$ from any configuration. 

By~\theoremref{thm:whp} with $\beta = 4$ we get that with probability 
  at least $1-1/n^4$ the algorithm stabilizes after $O(n \log^3 n)$ interactions.
Otherwise, with probability at most $1/n^4$ it ends up in some configuration 
  from where it takes at most $2n^4$ expected interactions to elect a leader.
The total expected number of steps is therefore also $O(n \log^3 n) + O(1) = O(n \log^3 n)$, 
  i.e. parallel time $O(\log^3 n)$.  
\end{proof}
\section{Experiments and Discussion}

\begin{figure}[t]
\centering
\includegraphics[scale=0.4]{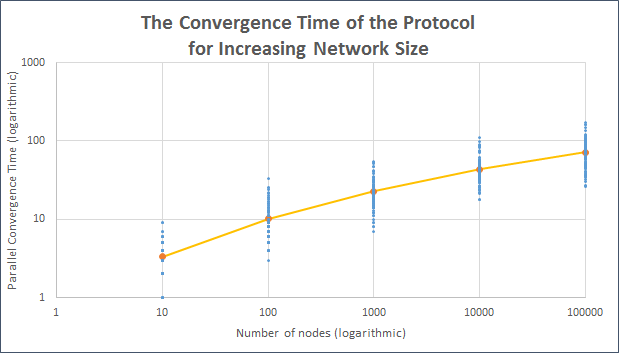}
\caption{The performance of the \emph{LM} protocol.  Both axes are logarithmic. The dots represent the results of individual experiments (100 for each network size), while the solid line represents the mean value for each  network size. 
}
\label{fig:graph}
\end{figure}

\paragraph{Empirical Data} 
We have also measured the stabilization time of our protocol for different network sizes. (Figure~\ref{fig:graph} presents the results in the form of a log-log plot.) 
The protocol stabilizes to a single leader quite fast, e.g., in less than 100 units of parallel time for a network of size $10^5$. This suggests that the constants hidden in the asymptotic analysis are small. The shape of the curve confirms the poly-logarithmic behavior of the protocol. 

\paragraph{Discussion} We have given the first population protocol to solve leader election in poly-logarithmic time, using a poly-logarithmic number of states per node. 
Together with the results of~\cite{AngluinAE2006}, the existence of our protocol implies that population protocols  can compute any semi-linear predicate on their input in time $O( n \log^5 n)$, with high probability, as long as memory per node is poly-logarithmic. 

Our result opens several avenues for future research. The first concerns \emph{lower bounds}. We conjecture that the lower bound for leader election in population protocols is $\Omega(\log n)$, irrespective of the number of states used by the protocol. Further, empirical data suggests that the analysis of our algorithm can be further tightened, cutting off logarithmic factors. It would also be interesting to prove a tight a trade-off between the amount of memory available per node and the running time of the protocol.


\bibliographystyle{alpha}
\bibliography{biblio}

\newcommand{\etalchar}[1]{$^{#1}$}
\begin{thebibliography}{AAD{\etalchar{+}}06}

\bibitem[AAD{\etalchar{+}}06]{AADFP06}
Dana Angluin, James Aspnes, Zo{\"e} Diamadi, Michael~J Fischer, and Ren{\'e}
  Peralta.
\newblock Computation in networks of passively mobile finite-state sensors.
\newblock {\em Distributed computing}, 18(4):235--253, 2006.

\bibitem[AAE06]{AngluinAE2006}
Dana Angluin, James Aspnes, and David Eisenstat.
\newblock Stably computable predicates are semilinear.
\newblock In {\em Proceedings of PODC 2006}, pages 292--299, 2006.

\bibitem[AAE08a]{AAE08le}
Dana Angluin, James Aspnes, and David Eisenstat.
\newblock Fast computation by population protocols with a leader.
\newblock {\em Distributed Computing}, 21(3):183--199, September 2008.

\bibitem[AAE08b]{AAE08}
Dana Angluin, James Aspnes, and David Eisenstat.
\newblock A simple population protocol for fast robust approximate majority.
\newblock {\em Distributed Computing}, 21(2):87--102, July 2008.

\bibitem[AAER07]{AngluinAER2007}
Dana Angluin, James Aspnes, David Eisenstat, and Eric Ruppert.
\newblock The computational power of population protocols.
\newblock {\em Distributed Computing}, 20(4):279--304, November 2007.

\bibitem[AAFJ06]{angluin2006self}
Dana Angluin, James Aspnes, Michael~J Fischer, and Hong Jiang.
\newblock Self-stabilizing population protocols.
\newblock In {\em Principles of Distributed Systems}, pages 103--117. Springer,
  2006.

\bibitem[AGV15]{AGVsubmitted}
Dan Alistarh, Rati Gelashvili, and Milan Vojnovic.
\newblock Fast and exact majority in population protocols.
\newblock Technical Report Available Online as MSR-TR-2015-13, 2015.
\newblock Submitted to PODC 2015.

\bibitem[CDS14]{chen2014deterministic}
Ho-Lin Chen, David Doty, and David Soloveichik.
\newblock Deterministic function computation with chemical reaction networks.
\newblock {\em Natural computing}, 13(4):517--534, 2014.

\bibitem[DH13]{doty2013leaderless}
David Doty and Monir Hajiaghayi.
\newblock Leaderless deterministic chemical reaction networks.
\newblock In {\em DNA Computing and Molecular Programming}, pages 46--60.
  Springer, 2013.

\bibitem[DS15]{DS15}
David Doty and David Soloveichik.
\newblock Stable leader election in population protocols requires linear time.
\newblock ArXiv preprint. http://arxiv.org/abs/1502.04246, 2015.

\bibitem[DV12]{DV12}
Moez Draief and Milan Vojnovic.
\newblock Convergence speed of binary interval consensus.
\newblock {\em SIAM Journal on Control and Optimization}, 50(3):1087--1109,
  2012.

\bibitem[EW13]{EW13}
Yuval Emek and Roger Wattenhofer.
\newblock Stone age distributed computing.
\newblock In {\em Proceedings of the 2013 ACM Symposium on Principles of
  Distributed Computing}, PODC '13, pages 137--146, New York, NY, USA, 2013.
  ACM.

\bibitem[FJ06]{fischer2006self}
Michael Fischer and Hong Jiang.
\newblock Self-stabilizing leader election in networks of finite-state
  anonymous agents.
\newblock In {\em Principles of Distributed Systems}, pages 395--409. Springer,
  2006.

\bibitem[MNRS14]{Spirakis}
George~B. Mertzios, Sotiris~E. Nikoletseas, Christoforos Raptopoulos, and
  Paul~G. Spirakis.
\newblock Determining majority in networks with local interactions and very
  small local memory.
\newblock In {\em Automata, Languages, and Programming - 41st International
  Colloquium, {ICALP} 2014, Copenhagen, Denmark, July 8-11, 2014, Proceedings,
  Part {I}}, ICALP '14, pages 871--882, 2014.

\bibitem[Neu66]{CA}
John~Von Neumann.
\newblock {\em Theory of Self-Reproducing Automata}.
\newblock University of Illinois Press, Champaign, IL, USA, 1966.

\bibitem[PVV09]{PVV09}
Etienne Perron, Dinkar Vasudevan, and Milan Vojnovic.
\newblock Using three states for binary consensus on complete graphs.
\newblock In {\em INFOCOM 2009, IEEE}, pages 2527--2535. IEEE, 2009.

\bibitem[SNY{\etalchar{+}}10]{sudo2010loosely}
Yuichi Sudo, Junya Nakamura, Yukiko Yamauchi, Fukuhito Ooshita, Hirotsugu
  Kakugawa, and Toshimitsu Masuzawa.
\newblock Loosely-stabilizing leader election in population protocol model.
\newblock In {\em Structural Information and Communication Complexity}, pages
  295--308. Springer, 2010.

\end{thebibliography}

\appendix

\end{document}